\documentclass[12pt,draftcls,onecolumn]{IEEEtran}
\usepackage{cite}
\usepackage{amsmath,amssymb}    
\usepackage[dvips]{graphicx}   
\usepackage{verbatim}   
\usepackage{color}      
\usepackage{subfigure}  
\usepackage{epsfig}
\usepackage{float}
\usepackage{algorithm}
\usepackage{algorithmic}
\usepackage{setspace}
\usepackage{subfigure}
\newtheorem{lemma}{\bf Lemma}
\newtheorem{theorem}{\bf Theorem}
\newtheorem{proposition}{\bf Proposition}
\newtheorem{definition}{\bf Definition}

\newenvironment{proof}[1][Proof]{\begin{trivlist}
\item[\hskip \labelsep {\bfseries #1}]}{\end{trivlist}}

\makeatletter
\newcommand*{\rom}[1]{\expandafter\@slowromancap\romannumeral #1@}
\makeatother
\begin{document}
\title{\textbf{Distributed Detection in Tree Networks: Byzantines and Mitigation Techniques}}
\author{Bhavya~Kailkhura,~\IEEEmembership{Student Member,~IEEE}, Swastik~Brahma,~\IEEEmembership{Member,~IEEE}, Berkan Dulek,~\IEEEmembership{Member,~IEEE}, Yunghsiang~S Han,~\IEEEmembership{Fellow,~IEEE},
Pramod~K.~Varshney,~\IEEEmembership{Fellow,~IEEE}
\thanks{This work was supported by the Center for Advanced Systems and Engineering at Syracuse University.}
\thanks{The authors would like to thank Aditya Vempaty for his valuable comments and suggestions to improve the
quality of the paper.}
\thanks{B. Kailkhura, S. Brahma and P. K. Varshney are with Department of EECS, Syracuse University, Syracuse, NY 13244. (email: bkailkhu@syr.edu; skbrahma@syr.edu; varshney@syr.edu)}
\thanks{B. Dulek is with Department of Electrical and Electronics Engineering, Hacettepe University, Beytepe Campus, 06800 Ankara, Turkey. (email: berkan@ee.hacettepe.edu.tr)}
\thanks{Y. S. Han is with EE Department, National Taiwan University of Science and Technology, Taiwan, R. O. C. (email: yshan@mail.ntust.edu.tw)}}
\maketitle
\begin{abstract} 
In this paper, the problem of distributed
detection in tree networks in the presence of Byzantines is considered. 
Closed form expressions for optimal attacking strategies that minimize the miss detection error exponent at the fusion center (FC) are obtained. We also look at the problem from the network designer's (FC's) perspective. We study the problem of designing optimal distributed detection parameters in a tree network in the presence of Byzantines. 
Next, we model the strategic interaction between the FC and the attacker as a Leader-Follower (Stackelberg) game. This formulation provides a methodology for predicting attacker and defender (FC) equilibrium strategies, which can be used to implement the optimal detector.
Finally, a reputation based scheme to identify Byzantines is proposed and its performance is analytically evaluated. We also provide some numerical examples to gain insights into the solution.
\end{abstract}

\section{Introduction}
Distributed detection deals with the problem of making a global decision regarding a phenomenon based on local decisions collected from several remotely located sensing nodes. Distributed detection research has traditionally focused on the parallel network topology, in which nodes directly transmit their observations or decisions to the Fusion Center (FC)~\cite{Varshney}\cite{Viswanathan}\cite{veer}. Despite its theoretical importance and analytical tractability, parallel topology  may not always reflect the
practical scenario. In certain cases, it may be required to place the nodes outside their communication range with the FC. Then, the coverage area can be increased by forming a multi-hop network,
where nodes are organized hierarchically into multiple levels
(tree networks). Some examples of tree networks include wireless sensor and military communication networks. For instance, the IEEE 802.15.4 (Zigbee) specifications \cite{Aliance} and IEEE 802.22b \cite{802_22bDraft} support tree networks.

Typically, a network embodies a large number of inexpensive sensors, which are deployed in an open environment to collect the observations regarding a certain phenomenon and, therefore, are susceptible to many kinds of attacks. A typical example is a Byzantine attack. While Byzantine attacks (originally proposed in \cite{Lamport}) may, in
general, refer to many types of malicious behavior, our focus in this paper is on data-falsification
attacks~\cite{vempaty_spm13,frag, Rifa, Marano, Rawat, Kailkhura2013, Kailkhura, aditya,a2,bhavyaj}, where an attacker sends false (erroneous) data to the FC to degrade detection performance. In this paper, we refer to such data falsification attackers as \textit{Byzantines}, and the data thus fabricated as \textit{Byzantine data}.

\subsection{Related Work}
Recently, distributed detection in the presence of Byzantine attacks has been explored in~\cite{Marano, Rawat}, where the problem of determining the most effective attacking strategy for the Byzantines was investigated. However, both works focused only on parallel topology. 
The problem considered in this paper is
most related to our earlier papers~\cite{bhavyaj, Kailkhura2013}. In~\cite{bhavyaj, Kailkhura2013}, we studied the problem of distributed detection in perfect tree networks (all intermediate nodes in the tree have the same number of children) with Byzantines under the assumption that the FC does not know which decision bit is sent from which node and assumes each received bit to originate from nodes at depth $k$ with a certain probability. Under this assumption, the attacker's aim was to maximize 
the false alarm probability for a fixed detection probability. When the number of nodes is large, by Stein's lemma~\cite{cover}, we know that the error exponent of the false alarm probability can be used as a surrogate for the false alarm probability. Thus, the optimal attacking strategy was obtained by making the error exponent of the false alarm probability at the FC equal to zero, which makes the decision fusion scheme to become completely incapable (blind). Some counter-measures were also proposed to protect
the network from such Byzantines.

There are several notable differences between this paper and our earlier papers~\cite{bhavyaj,Kailkhura2013}. 
First, in contrast to~\cite{bhavyaj,Kailkhura2013}, in this paper, the problem of distributed detection in regular tree networks\footnote{For a regular tree, intermediate nodes at different levels are allowed to have different degrees, i.e., number of children.
} with Byzantines is addressed in a practical setup where the FC has the knowledge of which bit is transmitted from which node. Note that, in practice, the FC knows which bit is transmitted from which node, e.g., using MAC schemes\footnote{In practice, one possible way to achieve this is by using the buffer-less TDMA MAC protocol, in which, distinct non-overlapping time slots are assigned (scheduled) to the nodes for communication. One practical example of such a scheme is given in~\cite{treeMAC}.}, and can utilize this information to improve system performance. 
Next, for the analysis of the optimal attack, we consider nodes residing at different levels of the tree to have different detection performance. We also allow Byzantines residing at different levels of the tree to have different attacking strategies and, therefore, provide a more general and comprehensive analysis of the problem as compared to~\cite{bhavyaj,Kailkhura2013}. We also study the problem from the network designer's perspective. 
Based on the information regarding which bit is transmitted from which node, we propose schemes to mitigate the effect of the Byzantines. 

\subsection{Main Contributions}
In this paper, it is assumed that the FC knows which bit is transmitted from which node.
Under this assumption, the problem of distributed detection in tree networks in the presence of Byzantines is considered.  
The main contributions of this paper are summarized below:
\begin{itemize}
\item Detection performance in tree networks with Byzantines is characterized in terms of the error exponent and a closed form expression for the optimal error exponent is derived.
\item The minimum attacking power required by the Byzantines to blind the FC in a tree network is obtained. It is shown that when more than a certain fraction of individual node decisions are falsified, the decision fusion scheme is completely jeopardized. 
\item The problem is also investigated from the network designer's perspective by focusing on the design of optimal distributed detection parameters in a tree network. 
\item We model the strategic interaction between the FC and the attacker as a Leader-Follower (Stackelberg) game and identify attacker and defender (FC) equilibrium strategies.
The knowledge of these equilibrium strategies can later be used to implement the optimal detector at the FC.
\item We propose a simple yet efficient reputation based scheme, which works even if the FC is blinded, to identify Byzantines in tree networks and analytically evaluate its performance.
\end{itemize}

The rest of the paper is organized as follows.
Section~\ref{sec2} introduces the system model.
In Section~\ref{sec3}, we study the problem from Byzantine's perspective and provide closed form expressions for optimal attacking strategies.
In Section~\ref{sec4}, we investigate the problem of designing optimal distributed detection parameters in the presence of Byzantines. 
In Section~\ref{sec5}, we model the strategic interaction between the FC and the attacker as a Leader-Follower
(Stackelberg) game and find equilibrium strategies.
In Section~\ref{sec6}, we introduce an efficient Byzantine identification scheme and analyze its performance.
Finally, Section~\ref{sec6} concludes the paper.

\section{System Model}
\label{sec2}
\begin{figure}[t]
  \centering
    \includegraphics[height=2.5in, width=!]{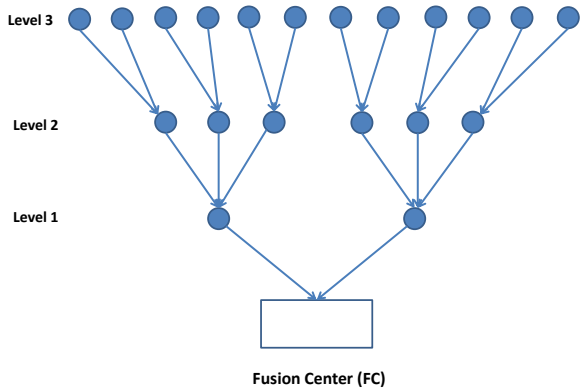}
    \vspace*{-0.3in}
    \caption{A distributed detection system organized as a regular tree $(a_1=2,\;a_2=3,\;a_3=2)$ is shown as an example.}\label{syst}
\vspace*{-0.1in}
\end{figure}
We consider a distributed detection system organized as a regular tree network rooted at the FC (See Figure~\ref{syst}). 
For a regular tree, all the leaf nodes are at the same level (or depth) and all the intermediate nodes at level $k$ have degree $a_k$. The regular tree is assumed to have a set $\mathcal{N}=\{\mathbb{N}_{k}\}_{k=1}^{K}$ of transceiver nodes, where $|\mathbb{N}_{k}|=N_{k}$ is the total number of nodes at level $k$.
We assume that the depth of the tree is $K>1$ and $a_k\geq 2$. The total number of nodes in the network is denoted as $N=\sum_{k=1}^K N_{k}$ and $\mathcal{B}=\{\mathbb{B}_{k}\}_{k=1}^{K}$ denotes the set of Byzantine nodes with $|\mathbb{B}_{k}|=B_{k}$, where $\mathbb{B}_{k}$ is the set of Byzantines at level $k$. The set containing the number of Byzantines residing at each level $k$, $1\leq k\leq K$, is referred to as an attack configuration, i.e., $\{B_k\}_{k=1}^{K}=\{|\mathbb{B}_{k}|\}_{k=1}^{K}$. 
Next, we define the \textit{modus operandi} of the nodes.

\subsection{Modus Operandi of the Nodes}
\label{Network}

We consider a binary hypothesis testing problem with two hypotheses $H_{0}$ (signal is absent)
and $H_{1}$ (signal is present). Under each hypothesis, it is assumed that the observations $Y_{k,i}$ at each node $i$ at level $k$ are conditionally independent. 
Each node $i$ at level $k$ acts as a source in the sense that it makes a one-bit (binary) local decision $v_{k,i} \in \{ 0,1\}$ regarding the absence or presence of the signal
using the likelihood ratio test (LRT) \footnote{Notice that, under the conditional independence assumption, the optimal
decision rule at the local sensor is a likelihood-ratio test~\cite{optimallrt}.}
    \begin{equation}
        \label{eqn1}
        \dfrac{p_{Y_{k,i}}^{(1)}(y_{k,i})}{p_{Y_{k,i}}^{(0)}(y_{k,i})} \quad \mathop{\stackrel{v_{k,i} = 1}{\gtrless}}_{v_{k,i} = 0} \quad \lambda_k,
    \end{equation}
where $\lambda_k$ is the threshold used at level $k$ (it is assumed that all the nodes at level $k$ use the same threshold $\lambda_k$) and $p_{Y_{k,i}}^{(j)} (y_{k,i})$ is the conditional probability density function (PDF) of observation $y_{k,i}$ under hypothesis $H_j$ for $j \in \{ 0,1\}$. We denote the probabilities of detection and false alarm of a node at level $k$ by $P_{d}^k=P(v_{k,i}=1|H_{1})$ and $P_{fa}^k=P(v_{k,i}=1|H_{0})$, respectively, which are functions of $\lambda_k$ and hold for both Byzantines and honest nodes.
After making its one-bit local decision $v_{k,i} \in \{ 0,1\}$, node $i$ at level $k$ sends $u_{k,i}$ to its parent node at level $k-1$, where $u_{k,i} = v_{k,i}$ if $i$ is an honest node, but for a Byzantine node $i$,
$u_{k,i}$ need not be equal to $v_{k,i}$. 
Node $i$ at level $k$ also receives the decisions $u_{k',j}$ of all successors $j$ at levels $k' \in [k+1,K]$, which are forwarded to node $i$ by its immediate children, and
forwards\footnote{For example, IEEE 802.16j mandates tree forwarding and IEEE 802.11s standardizes a tree-based routing protocol.} them to its parent node at level $k-1$. We assume error-free communication between children and the parent nodes.
Next, we present a mathematical model for the Byzantine attack.

\subsection{Byzantine Attack Model}
We define the following strategies $P_{j,1}^{H}(k)$, $P_{j,0}^H(k)$ and $P_{j,1}^B(k)$, $P_{j,0}^B(k)$ ($j \in \{0,1\}$ and $k=1,\cdots,K$) for the honest and Byzantine nodes at level $k$, respectively:\\
Honest nodes:
\begin{equation}
P_{1,1}^H(k)=1-P_{0,1}^H(k)=P_k^{H}(x=1|y=1)=1
\end{equation}
\begin{equation}
P_{1,0}^H(k)=1-P_{0,0}^H(k)=P_k^{H}(x=1|y=0)=0
\end{equation}
				
\noindent
Byzantine nodes:
\begin{equation}
P_{1,1}^B(k)=1-P_{0,1}^B(k)=P_k^{B}(x=1|y=1) 
\end{equation}
\begin{equation}
P_{1,0}^B(k)=1-P_{0,0}^B(k)=P_k^{B}(x=1|y=0) 
\end{equation}
where $P_k(x=a|y=b)$ is the conditional probability that a node at level $k$ sends $a$ to its parent when it receives $b$ from its child or its actual decision is $b$. For  notational convenience, we use $(P_{1,0}^k,P_{0,1}^k)$ to denote the flipping probability of the Byzantine node at level $k$.
Furthermore, we assume that if a node (at any level) is a Byzantine, then none of its ancestors and successors are Byzantine (non-overlapping attack configuration); otherwise, the effect of a Byzantine due to other Byzantines on the same path may be nullified  (e.g., Byzantine ancestor re-flipping the already flipped decisions of its successors). This means that every path from a leaf node to the FC will have at most one Byzantine. 
Notice that, for the attack configuration $\{B_k\}_{k=1}^{K}$, the total number of corrupted paths (i.e., paths containing a Byzantine node) from level $k$ to the FC are $\sum_{i=1}^{k}B_i\frac{N_k}{N_i}$, where $B_i\frac{N_k}{N_i}$ is the total number of nodes covered\footnote{Node $i$ at level $k'$ covers all its children at levels $k'+1$ to $K$ and itself.} at level $k$ by the presence of $B_i$ Byzantines at level 
$i$. If we denote $\alpha_k=\frac{B_k}{N_k}$, then, $\frac{\sum_{i=1}^{k}B_i\frac{N_k}{N_i}}{N_k}=\sum_{i=1}^{k}\alpha_i$ is the fraction of decisions coming from level $k$ that encounter a Byzantine along the way to the FC. We also approximate the probability that the FC receives the flipped decision $\bar{x}$ of a given node at level $k$ when its actual decision is $x$ as
$\beta_{\bar{x},x}^k=\sum_{j=1}^{k}\alpha_j P_{\bar{x},x}^{j},\;x\in\{0,1\}$.

\subsection{Binary Hypothesis Testing at the Fusion Center}
We consider the distributed detection problem under the Neyman-Pearson (NP) criterion. The FC receives decision vectors, $[\mathbf{z_1},\cdots,\mathbf{z_K}]$, where $\mathbf{z_k}$ for $k\in\{1,\cdots,K\}$ is a decision vector with its elements being $z_1,\cdots,z_{N_k}$, from the nodes at different levels of the tree. Then the FC makes the global decision about the phenomenon by employing the LRT. Due to system vulnerabilities, some of the nodes may be captured by the attacker and reprogrammed to transmit false information to the FC to degrade detection performance. 
We assume that the only information available at the FC is the probability $\beta_{\bar{x},x}^k$, which is the probability with which the data coming from level $k$ has been falsified.
Using this information, the FC calculates the probabilities $\pi_{j,0}^k=P(z_{i}=j|H_{0},k)$ and $\pi_{j,1}^k=P(z_{i}=j|H_{1},k)$, which are the distributions of received decisions $z_{i}$ originating from level $k$ and arriving to the FC under hypotheses $H_0$ and $H_1$. The FC makes its decision regarding the absence or presence of the signal using the following likelihood
ratio test 
\begin{equation}
\label{LRT}
\prod\limits_{k=1}^{K} \left(\frac{\pi_{1,1}^k}{\pi_{1,0}^k}\right)^{s_k} \left(\frac{1-\pi_{1,1}^k}{1-\pi_{1,0}^k}\right)^{N_k-s_k} \quad \mathop{\stackrel{H_1}{\gtrless}}_{H_0} \quad \eta
\end{equation}
where $s_k$ is the number of decisions that are equal to one and originated from level $k$, and the threshold $\eta$ is chosen in order to minimize the missed detection probability $(P_M)$ while keeping the false alarm probability $(P_F)$ below a fixed value $\delta$.\footnote{This type of problem setup is important, for instance, in Cognitive Radio Networks (CRN). In order to coexist with the primary user (PU), secondary users (SUs) must guarantee that their transmissions will not interfere with the transmission of the PU who have higher priority to access the spectrum.
}

Next, we derive a closed form expression for the optimal missed detection error exponent for tree networks in the presence of Byzantines, which will later be used as a surrogate for the probability of missed detection.

\begin{proposition}
For a $K$ level tree network employing the detection scheme as given in~\eqref{LRT}, the asymptotic detection performance can be characterized using the missed detection error exponent given below
\begin{equation}
\label{kld}
D=\sum_{k=1}^{K} {N_k} \left[\sum_{j \in \{0,1\}} \pi_{j,0}^k\log \dfrac{\pi_{j,0}^k}{\pi_{j,1}^k}\right].
\end{equation}
\end{proposition}
\begin{proof}
Let $\mathbf{Z}=[\mathbf{Z}_1,\cdots,\mathbf{Z}_{N_1}]$ denote the received decision vectors from the nodes at level $1$, where $\mathbf{Z_i}$ is the decision vector forwarded by the node $i$ at level $1$ to the FC. 
Observe that, $\mathbf{Z}_i$ for $i=1$ to $N_1$ are independent and identically distributed (i.i.d.). 
Therefore, using Stein's lemma~\cite{cover}, the optimal error exponent for the detection scheme as given in~\eqref{LRT} is the Kullback-Leibler divergence (KLD) \cite{Kullback} between the distributions $P(\mathbf{Z}|H_{0})$ and $P(\mathbf{Z}|H_{1})$. 
Summation term in~\eqref{kld} follows from the additive property of the KLD for independent distributions.
\end{proof}
Note that,~\eqref{kld} can be compactly written as $\sum_{k=1}^{K} {N_k} D_k(\pi_{j,1}^k||\pi_{j,0}^k)$ with $D_k(\pi_{j,1}^k||\pi_{j,0}^k)$ being the KLD between the data coming from node $i$ at level $k$ under $H_{0}$ and $H_{1}$. 
The FC wants to maximize the detection performance, while, the Byzantine attacker wants to degrade the detection performance as much as possible which can be achieved by maximizing and minimizing the KLD, respectively. 
Next, we explore the optimal attacking strategies for the Byzantines that degrade the detection performance most by minimizing the KLD.

\section{Optimal Byzantine Attack}
\label{sec3}
As discussed earlier, the Byzantines attempt to make the KL divergence as small as possible. Since the KLD is always non-negative, Byzantines attempt to choose $P(z_i=j|H_{0},k)$ and $P(z_i=j|H_{1},k)$ such that $D_k=0,\;\forall k$. In this case, an adversary can make the data that the FC receives from the nodes such that no information is conveyed from them. This is possible when
\begin{equation}
P(z_{i}=j|H_{0},k)=P(z_{i}=j|H_{1},k)\qquad \forall j\in \{0,1\},\;\forall k.
\label{eq10}
\end{equation}
Notice that, $\pi_{j,0}^k=P(z_{i}=j|H_{0},k)$ and $\pi_{j,1}^k=P(z_{i}=j|H_{1},k)$ can be expressed as
\begin{eqnarray}
&&
\pi_{1,0}^k
= \beta_{1,0}^k (1-P_{fa}^k)+(1-\beta_{0,1}^k) P_{fa}^k\label{eq3}\\
&& \pi_{1,1}^k 
= \beta_{1,0}^k (1-P_{d}^k)+(1-\beta_{0,1}^k) P_{d}^k.\label{eq4}
\end{eqnarray}
with $\beta_{1,0}^k=\sum_{j=1}^{k}\alpha_j P_{1,0}^j$ and $\beta_{0,1}^k=\sum_{j=1}^{k}\alpha_j P_{0,1}^j$.  Substituting \eqref{eq3} and \eqref{eq4} in \eqref{eq10} and after simplification, the condition to make the $D=0$ for a $K$-level network becomes $\sum_{j=1}^{k}\alpha_j(P_{1,0}^j+P_{0,1}^j)=1,\;\forall k$.
%
Notice that, when $\sum_{j=1}^{k}\alpha_j<0.5$, there does not exist any attacking probability distribution $(P_{0,1}^j,P_{1,0}^j)$ that can make $D_k = 0$, and, therefore, the KLD cannot be made zero. In the case of $\sum_{j=1}^{k}\alpha_j=0.5$, there exists a unique solution $(P_{0,0}^j,P_{1,0}^j)=(1,1)$, $\forall j$ that can make $D_k = 0$, $\forall k$. For the $\sum_{j=1}^{k}\alpha_j>0.5$ case, there exist infinitely many attacking probability distributions $(P_{0,1}^j,P_{1,0}^j)$ which can make $D_k = 0$, $\forall k$.
\begin{lemma}
In a tree network with $K$ levels, the minimum number of Byzantines needed to blind the FC (or to make $D_k=0,\;\forall k$) is given by $  B_1=\left \lceil\frac{N_1}{2} \right \rceil$.
\end{lemma}
\begin{proof}
The proof follows from the fact that the condition $\sum_{j=1}^{k}\alpha_j=0.5,\;\forall k$, is equivalent to $\alpha_1=0.5,\;\alpha_k=0,\;\forall k=2,\cdots,K$. 
\end{proof}
Next, we explore the optimal attacking probability distribution $(P_{0,1}^{k}, P_{1,0}^{k})$ that minimizes $D_k$ when $\sum_{j=1}^{k}\alpha_j < 0.5$, i.e., in the case where the attacker cannot make $D=0$. 
To analyze the problem, first we investigate the properties of $D_k$ with respect to $(P_{0,1}^{k}, P_{1,0}^{k})$ assuming $(P_{0,1}^{j}, P_{1,0}^{j}),\;1\leq j\leq k-1$ to be fixed. We show that attacking with symmetric flipping probabilities is the optimal strategy in the region where the attacker cannot make $D_k=0$. In other words, attacking with $P_{1,0}^k=P_{0,1}^k$ is the optimal strategy for the Byzantines.
\begin{lemma}
\label{lem}
In the region where the attacker cannot make $D_k=0$, i.e., for $\sum_{j=1}^{k}\alpha_j<0.5$, the optimal attacking strategy comprises of symmetric flipping probabilities $(P_{0,1}^{k}= P_{1,0}^{k}=p)$. In other words, any non zero deviation $\epsilon_i\in(0,p]$ in flipping probabilities $(P_{0,1}^{k}, P_{1,0}^{k})=(p-\epsilon_1,p-\epsilon_2)$, where $\epsilon_1\neq \epsilon_2$, will result in an increase in $D_k$.
\end{lemma}
\begin{proof}
Please see Appendix~\ref{ap1}.
\end{proof} 

In the next theorem, we present the solution for the optimal attacking probability distribution $(P_{j,1}^{k}, P_{j,0}^{k})$ that minimizes $D_k$ in the region where the attacker cannot make $D_k=0$.
\begin{theorem}
\label{th3}
In the region where the attacker cannot make $D_k=0$, i.e., for $\sum_{j=1}^{k}\alpha_j<0.5$, the optimal attacking strategy is given by $(P_{0,1}^{k}, P_{1,0}^{k})=(1,1)$.
\end{theorem}
\begin{proof}
Observe that, in the region where the attacker cannot make $D_k=0$, the optimal strategy comprises of symmetric flipping probabilities $(P_{0,1}^{k}= P_{1,0}^{k}=p)$. The proof is complete if we show that $D_k$ is a monotonically decreasing function of the flipping probability $p$.

After plugging in $(P_{0,1}^{k}, P_{1,0}^{k})=(p,p)$ in \eqref{eq3} and \eqref{eq4}, we get
\begin{eqnarray}
\pi_{1,1}^k&=&[\beta_{1,0}^{k-1}(1-P_d^{k})+(1-\beta_{0,1}^{k-1})P_d^{k}]+[\alpha_k(p-P_d^k(2p))+P_d^k]\\
\pi_{1,0}^k&=&[\beta_{1,0}^{k-1}(1-P_{fa}^{k})+(1-\beta_{0,1}^{k-1})P_{fa}^{k}]+[\alpha_k(p-P_{fa}^k(2p))+P_{fa}^k].
\end{eqnarray}
Now we show that $D_k$ is a monotonically  decreasing function of the parameter $p$ or in other words, $\dfrac{dD_k}{dp}<0$. After plugging in $\pi_{1,1}^{k'}=\alpha_k(1-2P_d^k)$ and $\pi_{1,0}^{k'}=\alpha_k(1-2P_{fa}^k)$ in the expression of $\dfrac{dD_k}{dp}$ and rearranging the terms, the condition $\dfrac{dD_k}{dp}<0$ becomes 
\begin{eqnarray}
\label{log2}
(1-2P_{d}^k)\left(\dfrac{1-\pi_{1,0}^k}{1-\pi_{1,1}^k}-\dfrac{\pi_{1,0}^k}{\pi_{1,1}^k}\right)
+(1-2P_{fa}^k)\log\left(\dfrac{1-\pi_{1,1}^k}{1-\pi_{1,0}^k}\dfrac{\pi_{1,0}^k}{\pi_{1,1}^k}\right)<0
\end{eqnarray}
Since $P_d^k>P_{fa}^k$ and $\beta_{\bar{x},x}^k<0.5$, we have $\pi_{1,1}^k>\pi_{1,0}^k$. Now, using the fact that $\dfrac{1-P_{d}^k}{1-P_{fa}^k}>\dfrac{1-2P_{d}^k}{1-2P_{fa}^k}$ and \eqref{eq-7}, we have
\begin{small} 
\begin{eqnarray}
&&
\dfrac{1-2P_{d}^k}{1-2P_{fa}^k}\left[\dfrac{1-\pi_{1,0}^k}{1-\pi_{1,1}^k}-\dfrac{\pi_{1,0}^k}{\pi_{1,1}^k}\right]<(\pi_{1,1}^k-\pi_{1,0}^k)\left[\dfrac{1}{\pi_{1,1}^k}+\dfrac{1}{1-\pi_{1,0}^k}\right]\\
&\Leftrightarrow&
\dfrac{1-2P_{d}^k}{1-2P_{fa}^k}\left[\dfrac{1-\pi_{1,0}^k}{1-\pi_{1,1}^k}-\dfrac{\pi_{1,0}^k}{\pi_{1,1}^k}\right]+\left[\dfrac{\pi_{1,0}^k}{\pi_{1,1}^k}-1\right]<
1-\dfrac{1-\pi_{1,1}^k}{1-\pi_{1,0}^k}.\label{tin}
\end{eqnarray}
\end{small}
Applying the logarithm inequality $(x-1)\geq\log x \geq \dfrac{x-1}{x}$, for $x>0$ to~\eqref{tin}, one can prove that \eqref{log2} is true.
\end{proof}
Next, to gain insights into the solution, we present some numerical results in Figure~\ref{dvsp}.
We plot $D_k$ as a function of the flipping probabilities $(P_{1,0}^k,\;P_{0,1}^k)$. We assume that the probability of detection is $P_d^k=0.8$, the probability of false alarm is $P_{fa}^k=0.2$, and the probability that the bit coming from level $k$ encounters a Byzantine is $\sum_{j=1}^{k}\alpha_j = 0.4$. 
We also assume that $P_{0,1}^{k}=P_{0,1}$ and $P_{1,0}^{k}=P_{1,0},\forall k$.
It can be seen that the optimal attacking strategy comprises of symmetric flipping probabilities and is given by $(P_{0,1}^{k}, P_{1,0}^{k})=(1,1)$, which corroborates our theoretical result presented in Lemma~\ref{lem} and Theorem~\ref{th3}.

\begin{figure}[t]
  \centering
    \includegraphics[height=2.5in, width=!]{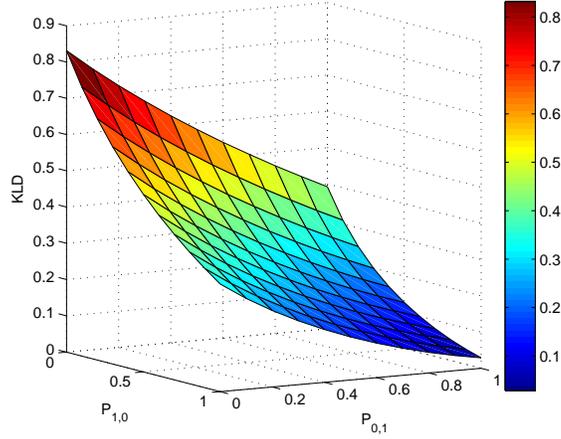}
    \vspace*{-0.1in}
    \caption{KLD $D_k$ vs. flipping probabilities when $P_d^k = 0.8$, 
$P_{fa}^k = 0.2$, and the probability that the bit coming from level $k$ encounters a Byzantine is $\sum_{j=1}^{k}\alpha_j = 0.4$.}\label{dvsp}
  \vspace*{-0.15in}
\end{figure}

We have shown that, for all $k$, 
\begin{equation}
D_k(P_{0,1}^{k}, P_{1,0}^{k})\geq D_k(1,1).\label{ho}
\end{equation} 
Now, by multiplying both sides of \eqref{ho} by $N_k$ and summing it over all $K$ we can show that the KLD, $D$, is minimized by $(P_{0,1}^{k}, P_{1,0}^{k})=(1,1)$, for all $k$, in the region $\sum\limits_{k=1}^{K}\alpha_k<0.5$.

Now, we explore some properties of $D_k$ with respect to $\sum_{j=1}^{k}\alpha_j$ in the region where the attacker cannot make $D_k=0$, i.e., for $\sum_{j=1}^{k}\alpha_j<0.5$. This analysis will later be used in exploring the problem from the network designer's perspective.

\begin{lemma}
\label{equiv}
$D_k^*$ =$\underset{(P_{j,1}^{k}, P_{j,0}^{k})}{\text{min}} D_k(\pi_{j,1}^k||\pi_{j,0}^k)$ is a continuous, decreasing and convex function of $\sum_{j=1}^{k}\alpha_j$ for  $\sum_{j=1}^{k}\alpha_j<0.5$.
\end{lemma}
\begin{proof}
The continuity of $D_k(\pi_{j,1}^k||\pi_{j,0}^k)$ with respect to the involved distributions implies the continuity of $D_k^*$.
To show that $D_k^*$ is a decreasing function of $t=\sum_{j=1}^{k}\alpha_j$, we use the fact that $\underset{(P_{0,1}^{k}, P_{1,0}^{k})}{\text{arg min}} D_k(\pi_{j,1}^k||\pi_{j,0}^k)$ is equal to $(1,1)$ for $\sum_{j=1}^{k}\alpha_j<0.5$ (as shown in Theorem~\ref{th3}). After plugging $(P_{0,1}^{k}, P_{1,0}^{k})=(1,1),\;\forall k,$ in the KLD expression, it can be shown that  $\dfrac{dD_k}{dt}<0$. Hence, $D_k^*$ is a monotonically decreasing function of $\sum_{j=1}^{k}\alpha_j$ for $\sum_{j=1}^{k}\alpha_j<0.5$. The convexity of $D_k^*$ follows from the fact that $D_k^*(\pi_{j,1}^k||\pi_{j,0}^k)$ is convex in $\pi_{j,1}^k$ and $\pi_{j,0}^k$, which are affine transformations of $\sum_{j=1}^{k}\alpha_j$ (Note that, convexity holds under affine transformation). 
\end{proof}

It is worth noting that Lemma~\ref{equiv} suggests that minimization/maximization of $\sum_{j=1}^{k}\alpha_j$ is equivalent to minimization/maximization of $D_k$. Using this fact, one can consider the probability that the bit coming from level $k$ encounters a Byzantine (i.e., $ t=\sum_{j=1}^{k}\alpha_j$) in lieu of  $D_k$ for optimizing the system performance. Observe that, the expression $t=\sum_{j=1}^{k}\alpha_j$ is much more tractable than the expression for $D_k$.

Next, to gain insights into the solution, we present some numerical results in Figure~\ref{dvst}. We plot $\underset{(P_{j,1}^{k}, P_{j,0}^{k})}{\text{min}} D_k$ as a function of the probability that the bit coming from level $k$ encounters a Byzantine, i.e., $t$. We assume that the probabilities of detection and false alarm are $P_d^k=0.8$ and $P_{fa}^k=0.2$, respectively. Notice that, when $t=0.5$, $D_k$ between the two probability distributions becomes zero. It is seen that $D_k^*$ is a continuous, decreasing and convex function of the fraction of covered nodes, $t$, for $t<0.5$, which corroborates our theoretical result presented in Lemma~\ref{equiv}.

\begin{figure}[t]
  \centering
    \includegraphics[height=2.5in, width=!]{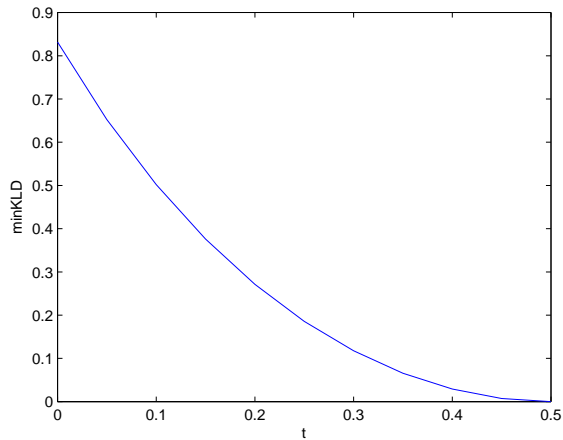}
    \vspace*{-0.1in}
    \caption{$\underset{(P_{j,1}^{k}, P_{j,0}^{k})}{\text{min}}$ $D_k$ vs probability that the bit coming from level $k$ encounters a Byzantine for $P_d^k = 0.8$ and $P_{fa}^k = 0.2$.}\label{dvst}
  \vspace*{-0.15in}
\end{figure}

Until now, we have explored the problem from the attacker's perspective. In the rest of the
paper, we look into the problem from a network designer's perspective and propose techniques
to mitigate the effect of Byzantines. First, we study the problem of designing optimal distributed detection parameters
in a tree network in the presence of Byzantines.

\section{System Design in the Presence of Byzantines}
\label{sec4}
For a fixed attack configuration $\{B_k\}_{k=1}^{K}$, the detection performance at the FC is a function of the local detectors used at the nodes in the tree network and the global detector used at the FC. This motivates us to study the problem of designing detectors, both at the nodes at different levels in a tree and at the FC, such that the detection performance is maximized. More specifically, we are interested in answering the question: How does the knowledge of the attack configuration $\{B_k\}_{k=1}^{K}$ affect the design of optimal distributed detection parameters?

By Stein's lemma~\cite{cover}, we know that in the NP setup for a fixed false alarm probability, the missed detection probability of the optimal detector can be minimized by maximizing the KLD. For an optimal detector at the FC, the problem of designing the local detectors can be formalized as follows:
\begin{equation}
\begin{aligned}
&  \underset{\{P_d^k,P_{fa}^k\}_{k=1}^K}{\text{max}}
& & \sum_{k=1}^{K} N_k \sum_{j \in \{0,1\}}P(z_{i}=j|H_{0},k)\log \dfrac{P(z_{i}=j|H_{0},k)}{P(z_{i}=j|H_{1},k)}. \label{p1}
 \end{aligned}
\end{equation}

The local detector design problem as given in~\eqref{p1} is a non-linear optimization problem. Furthermore, it is difficult to obtain a closed form solution for this problem. Also, observe that the solution space is not constrained to the likelihood ratio based tests. To solve the problem, we need to find the pairs $\{P_{d}^{k},P_{fa}^{k}\}_{k=1}^K$ which maximize the objective function as given in~\eqref{p1}. 
However, $P_{d}^{k}$ and $P_{fa}^{k}$ are coupled and, therefore, cannot be optimized independently. Thus, we first analyze the problem of maximizing the KLD for a fixed $P_{fa}^{k}$. We assume that $P_{fa}^{k}=y_k$ and $P_{d}^{k}=y_k+x_k$. Next, we analyze the properties of KLD with respect to $x_k$, i.e., $(P_{d}^{k}-P_{fa}^{k})$ in the region where attacker cannot blind the FC, i.e., for $\sum_{j=1}^{k}\alpha_j<0.5$, in order to study the local detector design problem. 
Notice that, in the region $\sum_{j=1}^{k}\alpha_j \geq 0.5$, $D_k=0$ and optimizing over local detectors does not improve the performance. 

\begin{lemma}
\label{lemma3}
For a fixed $P_{fa}^{k}=y_k$, when $\sum_{j=1}^{k}\alpha_j < 0.5$, the KLD, $D$, as given in~\eqref{kld} is a monotonically increasing function of $x_k=(P_{d}^{k}-P_{fa}^{k})$. 
\end{lemma}
\begin{proof}
 To prove this, we calculate the partial derivative of $D$ with respect to $x_k$. 
By substituting $P_{fa}^{k}=y_k$ and $P_{d}^{k}=y_k+x_k$ into~\eqref{kld}, the partial derivative of $D$ with respect to $x_k$ can be calculated as

\begin{eqnarray*}
&&
\dfrac{\partial D}{\partial x_k}=N_k\dfrac{\partial}{\partial x_k}\left[\pi_{1,0}^k\log\dfrac{\pi_{1,0}^k}{\pi_{1,1}^k}+(1-\pi_{1,0}^k)\log\dfrac{1-\pi_{1,0}^k}{1-\pi_{1,1}^k}\right]\\
&\Leftrightarrow&
\dfrac{\partial D}{\partial x_k}=N_k\pi_{1,1}^{k'}\left(\dfrac{1-\pi_{1,0}^k}{1-\pi_{1,1}^k}-\dfrac{\pi_{1,0}^k}{\pi_{1,1}^k}\right),
\end{eqnarray*}
where $\pi_{1,0}^k$ and $\pi_{1,1}^k$ are as given in \eqref{eq3} and \eqref{eq4}, respectively and $\pi_{1,1}^{k'}=(1-\beta_{0,1}^k-\beta_{1,0}^k)$. 
Notice that, 
\begin{eqnarray*}
&&
\left(\dfrac{1-\pi_{1,0}^k}{1-\pi_{1,1}^k}-\dfrac{\pi_{1,0}^k}{\pi_{1,1}^k}\right)>0\\
&\Leftrightarrow&
\pi_{1,1}^k>\pi_{1,0}^k.
\end{eqnarray*}
Thus, the condition to make $\dfrac{\partial D}{\partial x_k}>0$ simplifies to 
\begin{equation}
\label{c3}
\pi_{1,1}^{k'}>0 \Leftrightarrow 1>(
\beta_{0,1}^k+\beta_{1,0}^k)
\end{equation}
Substituting the values of $\beta_{1,0}^k$ and $\beta_{1,1}^k$,  the above condition can be written as:
\begin{eqnarray}
&&
\sum_{j=1}^{k}\alpha_j P_{1,0}^j+\sum_{j=1}^{k}\alpha_j P_{0,1}^j<1\\
&\Leftrightarrow&
\sum_{j=1}^{k}\alpha_j (P_{1,0}^j+P_{0,1}^j)<1
\end{eqnarray}
The above condition is true for any $0\leq P_{0,1}^j,P_{1,0}^j\leq 1$ when $\sum_{j=1}^{k}\alpha_j<0.5$. This completes the proof.
\end{proof}

Lemma~\ref{lemma3} suggests that one possible solution to maximize $D$ is to choose the largest possible $x_k$ constrained to $0 \leq x_k\leq 1-y_k$. The upper bound results from the fact that $\{P_{d}^{k},P_{fa}^{k}\}_{k=1}^{K}$ are probabilities and, thus, must be between zero and one. In other words, the solution is to maximize the probability of detection for a fixed value of probability of false alarm. In detection theory, it is well known that the likelihood ratio based test is optimum for this criterion. Thus, under the conditional independence assumption, likelihood ratio based test as given in~\eqref{LRT} is optimal for local nodes,
even in the presence of Byzantines, and the optimal operating points $\{P_{d}^{k*},P_{fa}^{k*}\}_{k=1}^{K}$ are independent of the Byzantines' parameters $\{\alpha_k\}_{k=1}^{K}$. 

To summarize, optimal local detectors for distributed detection in tree networks are likelihood ratio based detectors and are independent of the Byzantines' parameter $\{\alpha_k\}_{k=1}^{K}$. 
We further explore the problem from the network designer's (FC) perspective. In our previous analysis, we have assumed that the attack configuration $\{B_k\}_{k=1}^K$ is known and shown that the optimal local detector is independent of $\{\alpha_k\}_{k=1}^{K}$. However, notice that the KLD is the exponential decay rate of the error probability of the “optimal detector”. In other words, while optimizing over KLD, we implicitly assumed that the optimal detector, which is a likelihood ratio based detector, is used at the FC. Taking logarithm on both sides of \eqref{LRT}, the optimal decision rule simplifies to
\begin{equation}
\label{weight}
\sum_{k=1}^{K}[a_{1}^{k}s_k+a_{0}^{k}(N_k-s_k)]\quad \mathop{\stackrel{H_1}{\gtrless}}_{H_0} \quad \log\eta
\end{equation}
where the optimal weights are given by $a_{1}^{k}=\log\frac{\pi_{1,1}^k}{\pi_{1,0}^k}$ and $a_{0}^{k}=\log\frac{1-\pi_{1,1}^k}{1-\pi_{1,0}^k}$.
To implement the optimal detector, the FC needs to know the optimal weights $a_{j}^{k}$, which are functions of $\{\alpha_k\}_{k=1}^{K}$. 
In the next section, we are interested in answering the question: Is it possible for the FC to predict the attack configuration $\{B_k\}_{k=1}^K$ in the tree? The knowledge of this attack configuration can be used for
determining the optimal detector at the FC to improve the system performance. Notice that, learning/estimation based techniques can be used on data to determine the attack configuration. However, the FC has to acquire a large amount of data coming from the nodes over a long period of time to accurately estimate $\{B_k\}_{k=1}^K$. 

In the next section, we propose a novel technique to predict the attack configuration by considering the following scenario:
The FC, acting first, commits to a defensive strategy by deploying the defensive resources to protect the tree network, while the attacker chooses its best response or attack configuration after surveillance of this defensive strategy. Both, the FC and the Byzantines have to incur a cost to deploy the defensive resources and attack the nodes in the tree network, respectively.
We consider both the FC and the attacker to be strategic in nature and model the strategic interaction between them as a Leader-Follower (Stackelberg) game.
This formulation provides a framework for identifying attacker and defender (FC) equilibrium strategies, which can be used to implement the optimal detector. The main advantage of this technique is that the equilibrium strategies can be determined \textit{a priori} and, therefore, there is no need to observe a large amount of data coming from the nodes over a long period of time to accurately estimate $\{B_k\}_{k=1}^K$.

\section{Stackelberg Game for Attack Configuration Prediction Problems}
\label{sec5}
We model the strategic interaction between the FC and the attacker as a Leader-Follower (Stackelberg) game. 
We assume that the FC has to incur a cost for deploying the network and the Byzantine has to incur a cost\footnote{Due to variations in hardware complexity and the level of tamper-resistance present in nodes residing at different levels of the tree, the resources required to capture and tamper nodes at different levels  may be different and, therefore, nodes
have varying costs of being attacked.} for attacking the network.
It is assumed that the network designer or the FC has a cost budget $C_{budget}^{network}$ and the attacker has a cost budget $C_{budget}^{attacker}$\footnote{In this paper, we assume that the attacker budget $C_{budget}^{attacker}$ is such that $\sum\limits_{k=1}^{K}\alpha_k<0.5$, i.e., the attacker cannot make $D_k=0,\;\forall k$. Notice that, if the attacker can make $D_k=0$ for some $k=l$, then, it can also make $D_k=0,\;\forall k\geq l$. Also, $D_k=0$ implies that $\pi_{1,1}^k=\pi_{1,0}^k$ and, therefore, the weights $(a_1^k,a_0^k)$ in~\eqref{weight} are zero. In other words, the best the FC can do in the case when $D_k=0,\;\forall k\geq l$ is to ignore or discard the decisions of the nodes residing at level $k\geq l$. This scenario is equivalent to using the tree network with $(l-1)$ levels for distributed detection.}. 
More specifically, the FC wants to allocate the best subset of “defensive resources ” (denoted as  $\{\tilde{c}_k\}_{k=1}^K$)\footnote{Let $\tilde{c}_{k}$ denote the resources deployed or budget allocated by the FC to protect or deploy a node at level $k$.} from a set of available defensive resources $\mathbb{C}=(c_1,\cdots,c_n)$ (arranged in a descending order, i.e., $c_1\geq c_2\cdots\geq c_n$), where $n\geq K$, complying with its budget constraint $C_{budget}^{network}$ to different levels of the tree network. After the FC allocates the defensive resources or budget to different levels of the tree network, an attacker chooses an attack configuration, $\{B_k\}_{k=1}^K$ complying with his budget constraint $C_{budget}^{attacker}$ to maximally degrade the performance of the network.

Next, we formalize the Stackelberg game as a bi-level optimization problem.
For our problem, the upper level problem (ULP)
corresponds to the FC who is the leader of the game, while the lower level problem (LLP)
belongs to the attacker who is the follower.
\begin{equation}
\begin{split}
\underset{\{\tilde{c}_k\}_{k=1}^K\in\mathbb{C}}{\mathrm{maximize}}\quad& D(\{\tilde{c}_k\}_{k=1}^K) \\
\mbox{subject to}  \quad & \sum_{k=1}^{K}\tilde{c}_{k}N_k \leq C_{budget}^{network}\\
\, &  \underset{B_{k}\in \mathbb{Z}^{+}}{\text{minimize}} \quad D(\{B_k\}_{k=1}^K) \\
\, &  \text{subject to}\quad \sum_{k=1}^{K}\tilde{c}_{k}B_{k} \leq C_{budget}^{attacker} \\
\, & \qquad \qquad \quad 0 \leq B_{k} \leq N_k, \forall\, k = 1,2,\ldots,K
\end{split}
\end{equation}
where $\mathbb{Z}^{+}$ is the set of non-negative integers.
Notice that, the bi-level optimization problem, in general, is an NP-hard problem. In fact, the LLP is a variant of the packing formulation of the bounded knapsack problem with a non-linear objective function. This is, in general, NP-hard. Using existing algorithms, cost set $\{\tilde{c}_k\}_{k=1}^K$ and attack configuration $\{B_k\}_{k=1}^K$ can be determined at the cost of computational efficiency. In this paper, we identify a special case of the above problem which can be solved in polynomial time to determine the equilibrium strategies. 
Next, we discuss the relationships that
enable our problem to have a polynomial time solution.
We define profit $P(S)$ of an attack configuration $S=\{B_k\}_{k=1}^K$ as follows\footnote{In this section, we assume that the optimal operating point, i.e., $(P_{d}^{k*},P_{fa}^{k*})$, is the same for all the nodes in the tree network. It has been shown that the use of identical thresholds is asymptotically optimal for parallel networks~\cite{tsit}. We conjecture that this result is valid for tree networks as well and employ identical thresholds.} $$P(S)=D(\phi)-D(S)=D(\phi)-D(\{B_k\}_{k=1}^K),$$
where $D(\phi)$ is the KLD when there are no Byzantines in the network and $D(S)=D(\{B_k\}_{k=1}^K)$ is the KLD with $\{B_k\}_{k=1}^K$ Byzantines in the tree network. 
Next, we define the concept of dominance which will be used later to explore some useful properties of the optimal attack configuration $\{B_k\}_{k=1}^K$.

\begin{definition}
We say that a set $S_{1}$ dominates another set $S_{2}$ if
\begin{equation}
\label{no}
P(S_{1})\geq P(S_{2})\mbox{ and } C(S_{1})\leq C(S_{2}),
\end{equation}
where $P(S_{i})$ and $C(S_{i})$ denote the profit and cost incurred by using set $S_{i}$, respectively. If in~\eqref{no}, $P(S_{1})> P(S_{2})$, $S_1$ strictly dominates $S_2$ and if $P(S_{1})= P(S_{2})$, $S_1$ weakly dominates $S_2$.
\end{definition}

To solve the bi-level optimization problem, we first solve the LLP assuming the solution of the ULP to be some fixed $(\tilde{c}_1,\cdots,\tilde{c}_K)$. This approach will give us a structure of the optimal $\{B_k\}_{k=1}^K$ for any arbitrary $\{\tilde{c}_k\}_{k=1}^K$, which can later be utilized to solve the bi-level optimization problem.  
We refer to LLP as a maximum damage Byzantine attack problem. Observe that, knowing that the FC chooses $(\tilde{c}_1,\cdots,\tilde{c}_K)$, the LLP can be reformulated as follows: 
\begin{equation*}
\begin{aligned}
& \underset{\{B_{k}\}_{k=1}^{K}}{\text{minimize}}
& & \sum_{k=1}^{K} N_k D_k(\{B_{i}\}_{i=1}^{k}) \\
& \text{subject to}
& & \sum_{k=1}^{K}\tilde{c}_{k}B_{k} \leq C_{budget}^{attacker} \\
& & & 0 \leq B_{k} \leq N_{k}, \text{and integer}\; \forall k\\
\end{aligned}
\end{equation*}

Next, we discuss the relationships that
enable our maximum damage Byzantine attack problem to admit a polynomial time solution.

\subsection{Analysis of the Optimal Attack Configuration}

In this section, we identify a special case of the bounded knapsack problem (LLP) which can be solved in polynomial time. More specifically, we show that if the set of defensive resources $\mathbb{C}=(c_1,\cdots,c_n)$ satisfy the cost structure  $c_{max} \leq \left(\displaystyle\min_{k\in\{1,\cdots,K-1\}}\frac{N_{k+1}}{N_k}\right) \times c_{min}\footnote{Notice that, in the case of the perfect $M$-ary tree networks, the proposed cost structure simplifies to $c_{max} \leq M\times c_{min}$.}$ or $c_1\leq \displaystyle\min_k a_k\times c_n$, then, 
the optimal solution $\{B_k\}_{k=1}^K$ exhibits the properties given in the lemma below.

\begin{lemma}
\label{l5}
Given a $K$ level tree network with cost structure satisfying $c_{max} \leq \left(\displaystyle\min_{k\in\{1,\cdots,K-1\}}\frac{N_{k+1}}{N_k}\right) \times c_{min}$,
the best response of an attacker with cost budget $C_{budget}^{attacker}$ is $\{B_{k}\}_{k=1}^{K}$ with
\[  \begin{array}{lcl}
			B_{1}= \left \lfloor \frac{C_{budget}^{attacker}}{\tilde{c}_{1}}\right \rfloor \qquad \\
				\end{array}
\]
and the remaining elements of $B_k$  for $2 \leq k\leq K$ can be calculated recursively.
\end{lemma}
\begin{proof}
To prove Lemma~\ref{l5}, it is sufficient to show that:
\begin{enumerate}
\item KLD is a monotonically decreasing function of $B_k$, and,
\item Attacking parent nodes is a strictly dominant strategy.
\end{enumerate}

Lemma~\ref{equiv} suggests that the KLD is a monotonically decreasing function of $B_k$ in the region where attacker cannot make $D_k=0$ and, therefore, (1) is proved.
Next, we show that attacking parent nodes is a strictly dominant strategy. In other words, given a cost budget $C_{budget}^{attacker}$, it is more profitable for an attacker to attack the parent nodes.  
Observe that the KLD at level $k$ is a function of Byzantines' parameter $(B_1,\cdots,B_k)$. Thus, we denote it as $D_k(B_1,\cdots,B_k)$. 

In order to prove that attacking parent nodes is a strictly dominant strategy, it is sufficient to show that the attack configuration $S_1=(B_1,\cdots,B_j,B_{j+1},\cdots,B_K)$ strictly dominates the attack configuration $S_2=(B_1,\cdots,B_j-\delta,B_{j+1}+\delta\frac{N_{j+1}}{N_j},\cdots,B_K)$ for $\delta \in \{1,\cdots,B_j\}$. In other words, we want to show that $P(S_1)>P(S_2)$ and $C(S_1)\leq C(S_2)$. From the cost inequality it follows that $C(S_1)\leq C(S_2)$ because $c_{max} \leq (\displaystyle\min_k {N_{k+1}}/{N_k}) \times c_{min}\Rightarrow \tilde{c}_{j} \leq (\displaystyle {N_{j+1}}/{N_j}) \times \tilde{c}_{j+1}$.  
Also, note that if the attack configuration $S_1$ strictly dominates the attack configuration $S_2$, then, it will also strictly dominate any attack configuration $\tilde{S}_2$ with $\tilde{S}_2=(B_1,\cdots,B_j-\delta,B_{j+1}+\delta\gamma,\cdots,B_K)$, where $\gamma\leq \frac{N_{j+1}}{N_j}$.
Next, we show that $P(S_1)>P(S_2)$.

Since $D_j(B_1,\cdots,B_{j-1},B_j)<D_j(B_1,\cdots,B_{j-1},B_j-\delta)$, for $\delta\in\{1,\cdots,B_j\},\;\forall j$, it follows that
\begin{small}
\begin{eqnarray*}
&&
D_j(B_1,\cdots,B_{j-1},B_j)<D_j(B_1,\cdots,B_{j-1},B_j-\delta)\\
&\Leftrightarrow&
\sum_{k=1}^{j}D_k(B_1,\cdots,B_k)<\sum_{k=1}^{j-1}D_k(B_1,\cdots,B_k)+D_j(B_1,\cdots,B_{j-1},B_j-\delta)\\
&\Leftrightarrow&
\sum_{k=1}^{K}D_k(B_1,\cdots,B_k)<\sum_{k=1}^{j-1}D_k(B_1,\cdots,B_k)+D_j(B_1,\cdots,B_{j-1},B_j-\delta)\\
&&\qquad\qquad\qquad\qquad\qquad+\sum_{k=j+1}^{K}D_k(B_1,\cdots,B_j-\delta,B_{j+1}+\delta\frac{N_{j+1}}{N_j},B_{j+2},\cdots,B_k),
\end{eqnarray*}
\end{small}
where the last inequality follows from the fact that $\frac{B_j}{N_j}+\frac{B_{j+1}}{N_{j+1}}=\frac{B_j-\delta}{N_j}+\frac{B_{j+1}+\frac{N_{j+1}}{N_j}\delta}{N_{j+1}}$ and, therefore,
$$D_k(B_1,\cdots,B_j,B_{j+1},\cdots,B_k)=D_k(B_1,\cdots,B_j-\delta,B_{j+1}+\frac{N_{j+1}}{N_j}\delta,\cdots,B_k).$$

This implies that the set $S_1$ strictly dominates the set $S_2$.
From the results in Lemma~\ref{equiv}, it is seen that the profit is an increasing function of the attack nodes. Lemma~\ref{equiv} in conjunction with the fact that attacking parent nodes is a strictly dominant strategy implies Lemma~\ref{l5}.
\end{proof}
It can also be shown that the solution $\{B_k\}_{k=1}^{K}$ will be non-overlapping and unique under the condition that the attacker cannot make $D_k=0,\;\forall k$.
\subsection{Bi-Level Optimization Algorithm}
\label{case3}

Based on Lemma~\ref{l5}, in this section we will present a polynomial time algorithm to solve the bi-level optimization problem, i.e., to find $\{\tilde{c_{k}}\}_{k=1}^{K}$ and $\{B_{k}\}_{k=1}^{K}$.
Using the cost structure $c_{max} \leq \left(\displaystyle\min_k\frac{N_{k+1}}{N_k}\right) \times c_{min}$, the attack configuration $\{B_k\}_{k=1}^K$ as given in Lemma~\ref{l5} can be determined in a computationally efficient manner. 
Due to structure of the optimal $\{B_k\}_{k=1}^K$, the bi-level optimization problem simplifies to finding the solution $\{\tilde{c_k}\}_{k=1}^K$ of the ULP.

To solve this problem, we use an iterative elimination approach. We start by listing all $n \choose K$ combinations from the set $\mathbb{C}$, denoted as, $S=\{s_i\}_{i=1}^{n \choose K}$. Without loss of generality, we assume that the elements of $s_i=\{c_1^i,\cdots,c_K^i\}$ are arranged in  descending order, i.e., $c_k^i\geq c_{k+1}^i,\forall k$. 
Notice that, all these $n \choose K$ combinations will satisfy $c_k^i\leq \frac{N_{k+1}}{N_k} c_{k+1}^i$, because $$c_k^i\leq c_{max}\leq \min_j\frac{N_{j+1}}{N_j}c_{min}\leq \min_j\frac{N_{j+1}}{N_j}c_{k+1}^i\leq \frac{N_{k+1}}{N_k}c_{k+1}^i.$$
Next, we discard all those subsets $s_i$ from $S$ which violate the network designer's cost budget constraint. If the set $S$ is empty, then there does not exist any solution for the ULP. Otherwise, the problem reduces to finding the subset $s_i$ which maximizes the KLD. To find the subset $s_i$ which maximizes the KLD, using the dominance relationship we start with assigning the cost $\tilde{c_1}=\displaystyle\min_{k\in s} c_1^{k}$, where $s$ has the elements which are solutions of $\displaystyle\arg\min_i\left \lfloor \frac{C_{budget}^{attacker}}{c_1^{i}}\right \rfloor$. Next, we discard all those subsets $s_i$ from $S$ which do not have $\tilde{c_1}$ as their first element and solve the problem recursively.
 
Pseudo code of the polynomial time algorithm to find $\{\tilde{c_{k}}\}_{k=1}^{K}$ and $\{B_{k}\}_{k=1}^{K}$
is presented as Algorithm~\ref{findnumberofByzantine1}.

\begin{algorithm} [] \setstretch{1.75}
\caption{Bi-Level Optimization Algorithm }          
\label{findnumberofByzantine1}                           
\begin{algorithmic} [1] 
\REQUIRE{$\mathbb{C}=\{c_k\}_{k=1}^n$ with $c_{max} \leq \left(\displaystyle\min_j\frac{N_{j+1}}{N_j}\right) \times c_{min}$}
\STATE $S \leftarrow \text{All $K$ out of $n$ combinations} \;\{s_i\}_{i=1}^{n \choose K}$ \text{with elements of $s_i$ arranged in decreasing order}
\FOR{$i=1$ \TO ${n \choose K}$}
\IF {$\sum\limits_{k=1}^{K}c_k^i\times N_k > C_{budget}^{network}$}
\STATE $S \leftarrow S/s_i$ 
\ENDIF
\ENDFOR
\IF {$S$ is an empty set}
\RETURN $(\phi,\;\phi)$
\ELSE
\FOR{$k=1$ \TO $K$}
\STATE $\tilde{c}_k=\displaystyle\min_{j\in s} c_k^j$ \text{where} $s$ \text{has elements which are solutions of} $\displaystyle\arg\min_i\left \lfloor \frac{C_{budget}^{attacker}}{c_k^{i}}\right \rfloor$
\STATE $B_{k} \leftarrow \left \lfloor \frac{C_{budget}^{attacker}}{\tilde{c_k}}\right \rfloor$
\STATE $C_{budget}^{attacker} \leftarrow (C_{budget}^{attacker}-\tilde{c}_k B_k)$
\ENDFOR
\RETURN $(\{\tilde{c}_k\}_{k=1}^K,\;\{B_k\}_{k=1}^K)$
\ENDIF
\end{algorithmic}
\end{algorithm}

\subsection{An Illustrative Example}
\label{example}
 Let us consider a two-level network with $N_{1}=6$ and $N_{2}=12$. 
We assume that $\mathbb{C}=\{4,\;3,\;2\}$, $C_{budget}^{network}=60$ and $C_{budget}^{attacker}=11$. Next, we solve the bi-level optimization problem. 
Observe that, costs satisfy ${c}_{1} \leq 2\times {c}_{3}$.
So the algorithm chooses the solution of the ULP as ($\tilde{c}_{1} = 4$, $\tilde{c}_{2} = 3$) and the solution of the LLP as
($B_{1}=\left \lfloor \frac{11}{4}\right \rfloor =2$, $B_{2}=\left \lfloor\frac{11-2\times 4}{3}\right \rfloor =1$). 
To corroborate these result, in Figure~\ref{feas}, we plot the $\underset{P_{1,0},P_{0,1}}{\text{min}}$ KLD for all combinations of the parameters $B_1$ and $B_2$ in the tree. We vary the parameter $B_1$ from $0$ to $6$ and $B_2$ from $0$ to $12$. 
All the feasible solutions are plotted in red and unfeasible solutions are plotted in blue. 
Figure~\ref{feas} corroborates the results of our algorithm.
\begin{figure}[t]
  \centering
    \includegraphics[height=2.5in, width=!]{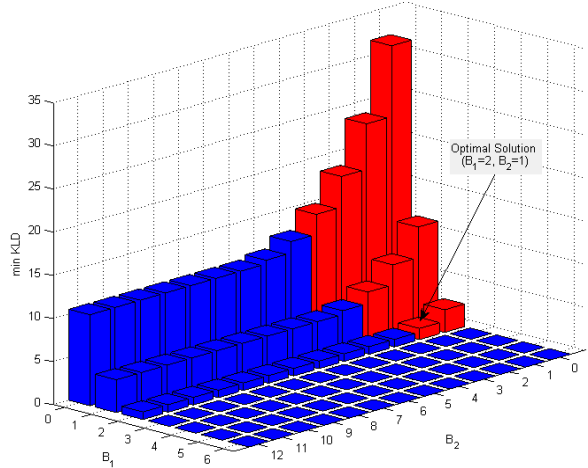}
    \vspace*{-0.1in}
    \caption{min KLD vs. attack configuration $(B_1,B_2)$ for $P_d = 0.9$, $P_{fa} = 0.1$.}\label{feas}
  \vspace*{-0.15in}
\end{figure}

Notice that, the attack configuration $\{B_k\}_{k=1}^{K}$ is the set containing
the \textit{number} of Byzantines residing at different levels of the tree. However, the FC cannot identify the Byzantines in the network. Also, notice that when the adversary attacks more than $50\%$ of nodes at level $1$, the decision fusion scheme becomes completely incapable. In these scenarios, where the FC is blind, the knowledge of attack configuration will not incur any performance benefit.
Next, we present a reputation-based Byzantine identification/mitigation scheme, which works even when the network is blind, in order to improve the detection performance of the network. We propose a simple yet efficient Byzantine identification scheme and analyze its performance.

\section{An Efficient Byzantine Identification Scheme}
\label{performance}
In this section, we propose and analyze a Byzantine identification scheme to be implemented at the FC. 

\subsection{Byzantine Identification Scheme}
We assume that the FC has the knowledge of the attack model and utilizes this knowledge to identify the Byzantines. The FC observes the local decisions of each node over a time window $T$, which can be denoted  by $(k,i)=[u_1(k,i),\ldots, u_T(k,i)]$ for $1 \leq i \leq N_k$ at level $1 \leq k \leq K$. 
We also assume that there is one honest anchor node with probability of detection $P_d^A$ and probability of false alarm $P_{fa}^A$ present and known to the FC. We employ the anchor node to provide the gold standard which is used to detect whether or not other nodes are Byzantines. The FC can also serve as an anchor node when it can directly observe the phenomenon and make a decision. We denote the Hamming distance between reports of the anchor node and an honest node $i$ at level $k$ over the time window $T$ by $d_{H}^{A}(k,i)=||U^A-U^H(k,i)||$, that is the number of elements that are different between  $U^A$ and $U^H(k,i)$. Similarly, the Hamming distance between reports of the anchor node and a Byzantine node $i$ at level $k$ over the time window $T$ is denoted by $d_{B}^{A}(k,i)=||U^A-U^B(k,i)||$.
Since the FC is aware of the fact
that Byzantines might be present in the network, it compares the Hamming distance of a node $i$ at level $k$ 
to a threshold $\eta_k$, $\forall i, \forall k$ (a procedure to calculate $\eta_k$ is discussed later in the paper), to make a decision to identify the Byzantines. In tree networks, a Byzantine node
alters its decision as well as received decisions
from its children prior to transmission in order to undermine the network performance. Therefore, solely based on the observed data of a node $i$ at level $k$, the FC cannot determine whether the data has been flipped by the node $i$ itself or by one of its Byzantine parent node. In our scheme, the FC makes the inference about a node being Byzantine by analyzing the data from the node $i$ as well as its predecessor nodes' data. FC starts from the nodes at level $1$ and computes the Hamming distance between reports of the anchor node and the nodes at level $1$. FC declares node $i$ at level $1$ to be a Byzantine if and only if the Hamming distance of node $i$ is greater than a fixed threshold $\eta_1$. Children of identified Byzantine nodes $\mathbb{C}(\mathbb{B}_1)$ are not tested further because of the non-overlapping condition. However, if a level $1$ node is determined not to be a Byzantine, then, the FC tests its children nodes at level $2$. The FC declares node $i$ at level $k$, for $2\leq k\leq K$, to be a Byzantine if and only if the Hamming distance of node $i$ is greater than a fixed threshold $\eta_k$ and Hamming distances of all predecessors of node $i$ is less than equal to their respective thresholds $\eta_j$. 

In this way, it is possible to counter the data falsification attack by isolating Byzantine nodes from the information fusion process.
The probability that a Byzantine node $i$ at level $k$ is isolated at the end of the time window $T$, is denoted as $P_{B}^{iso}(k,i)$. 

\subsection{Performance Analysis}
 
As aforementioned, local decisions of the nodes are compared to the decisions of the anchor node over a time window of length $T$. The probability that an \textit{honest} node $i$ at level $k$ makes a decision that is different from the anchor node is given by
\begin{eqnarray*}
&&P_{diff}^{AH}(k,i)\\
&=& P(u_i^A=1,u_{k,i}^H=0,H_0)+P(u_i^A=0,u_{k,i}^H=1,H_0)\\
&&+P(u_i^A=1,u_{k,i}^H=0,H_1)+P(u_i^A=0,u_{k,i}^H=1,H_1)\\
&=&P_0[(P_{fa}^k+P_{fa}^A)-2P_{fa}^k P_{fa}^A]+P_1[(P_{d}^k+P_d^A)-2P_{d}^k P_d^A]\\
&\doteq& P_0[P_{diff}^{AH}(k,i,0)]+P_1[P_{diff}^{AH}(k,i,1)]~.
\end{eqnarray*}
where the prior probabilities of the two hypotheses $H_0$ and $H_1$ are
denoted by $P_0$ and $P_1$, respectively.
The probability that a \text{Byzantine} node $i$ at level $k$ sends a decision different from that of the anchor node is given by
\begin{eqnarray*}
&&P_{diff}^{AB}(k,i)\\
&=& P(u_{i}^A=1,u_{k,i}^B=0,H_0)+P(u_i^A=0,u_{k,i}^B=1,H_0)\\
 &&+P(u_i^A=1,u_{k,i}^B=0,H_1)+P(u_i^A=0,u_{k,i}^B=1,H_1)\\
&=&P_0[P_{fa}^AP_{fa}^k +(1-P_{fa}^A)(1-P_{fa}^k)]+P_1[P_d^AP_d^k+(1-P_d^A)(1-P_d^k)]\\
&\doteq& P_0[P_{diff}^{AB}(k,i,0)]+P_1[P_{diff}^{AB}(k,i,1)]~.
\end{eqnarray*}

The difference between the reports of a node and the anchor node under hypothesis $l\in\{0,1\}$ (i.e., $d_{I}^{A}(k,i,l),\;I\in\{H,B\}$) is a Bernoulli random variable with mean $P_{diff}^{AH}(k,i,l)$ for honest nodes and $P_{diff}^{AB}(k,i,l)$ for Byzantines. FC declares node $i$ at level $k$ to be a Byzantine if and only if the Hamming distance of node $i$ is greater than a fixed threshold $\eta_k$ and Hamming distances of all predecessors of node $i$ are less than equal to their respective thresholds $\eta_j$.
The probability that a Byzantine node $i$ at level $k$ is isolated at the end of the time window $T$ can be expressed as

{\small \begin{eqnarray*}
&&
P_B^{iso}(k,i)=P[(d_{B}^{A}(k,i)>\eta_k),(d_{H}^{A}(k-1,i)\leq\eta_{k-1}),\cdots,(d_{H}^{A}(1,i)\leq\eta_{1})]\\
&=&\displaystyle \sum_{l\in\{0,1\}} P_l\left[
P[d_{B}^{A}(k,i,l)>\eta_k]\prod_{m=1}^{k-1}P[d_{H}^{A}(m,i,l)\leq\eta_{m}]\right]\\
&=& \displaystyle \sum_{l\in\{0,1\}} P_l\left[
\sum_{j=\eta_k+1}^{T} \binom {T} {j}(P_{diff}^{AB}(k,i,l))^j(1-P_{diff}^{AB}(k,i,l))^{T-j}\prod_{m=1}^{k-1}\left[\sum_{j=0}^{\eta_m} \binom {T} {j}(P_{diff}^{AH}(m,i,l))^j(1-P_{diff}^{AH}(m,i,l))^{T-j}\right]\right].
\end{eqnarray*}}
For large $T$, by using the normal approximation, we get
\begin{small}
\begin{equation*}
P_B^{iso}(k,i)=\displaystyle \sum_{l\in\{0,1\}} P_l\left[ Q\left(\frac{\eta_k-TP_{diff}^{AB}(k,i,l)}{\sqrt{(TP_{diff}^{AB}(k,i,l)(1-P_{diff}^{AB}(k,i,l)))}} \right)\prod_{m=1}^{k-1}Q\left(\frac{TP_{diff}^{AH}(m,i,l)-\eta_m}{\sqrt{(TP_{diff}^{AH}(m,i,l)(1-P_{diff}^{AH}(m,i,l)))}} \right)\right].
\end{equation*}
\end{small}
This can be written recursively as follows
\begin{equation}
\label{recur}
P_B^{iso}(k+1,i)=\displaystyle \sum_{l\in\{0,1\}} P_l\left[(1-b(k,l))\left(\frac{a(k+1,l)}{a(k,l)}\right)P_B^{iso}(k,i,l)\right],
\end{equation}
with $P_B^{iso}(k,i)\doteq\displaystyle \sum_{l\in\{0,1\}} P_l [P_B^{iso}(k,i,l)]$, and
$$a(k,l)=Q\left(\dfrac{\eta_{k}-TP_{diff}^{AB}(k,i,l)}{\sqrt{(TP_{diff}^{AB}(k,i,l)(1-P_{diff}^{AB}(k,i,l)))}} \right),$$
$$b(k,l)=Q\left(\dfrac{\eta_{k}-TP_{diff}^{AH}(k,i,l)}{\sqrt{(TP_{diff}^{AH}(k,i,l)(1-P_{diff}^{AH}(k,i,l)))}} \right).$$
One can choose $\eta_k$ such that the isolation probability of honest nodes at level $k$ based solely on its data under the hypothesis $H_l$ (i.e., $b(k,l)$) is constrained to some value $\delta_k<<0.5$. In other words, we choose $\eta_k$ such that $\displaystyle \max_{l\in\{0,1\}} b(k,l)=\delta_k$, i.e., 
\begin{small}
\begin{equation}
\label{thresholdp}
\eta_{k}=Q^{-1}(\delta_k){\sqrt{TP_{diff}^{AH}(k,i,l^*)(1-P_{diff}^{AH}(k,i,l^*))}} + TP_{diff}^{AH}(k,i,l^*)
\end{equation}
\end{small}
where $l^*=\displaystyle\arg\max_l b(k,l)$. Now, the expression for $a(k,l)$ can be written as
\begin{small}
\begin{equation*}
a(k,l)=Q\left(\dfrac{Q^{-1}(\delta_k){\sqrt{P_{diff}^{AH}(k,i,l^*)(1-P_{diff}^{AH}(k,i,l^*))}} + \sqrt{T}(P_{diff}^{AH}(k,i,l^*)-P_{diff}^{AB}(k,i,l))}{\sqrt{P_{diff}^{AB}(k,i,l)(1-P_{diff}^{AB}(k,i,l))}} \right)
\end{equation*}
\end{small}
Now using the fact that $\displaystyle \max_l P_{diff}^{AH}(k,i,l)<\displaystyle \min_l P_{diff}^{AB}(k,i,l)$, it can be shown that $(P_{diff}^{AH}(k,i,l^*)-P_{diff}^{AB}(k,i,l))<0,\;\forall i$ and, therefore, $\lim\limits_{T\rightarrow \infty} a(k,l)=1$.

\begin{lemma}
For a $K$ level tree network, for our proposed Byzantine identification scheme, the asymptotic (i.e., $T\rightarrow \infty$) probability that a Byzantine node $i$ at level $k+1$, for $1\leq k\leq K-1$, is isolated is lower-bounded by, $$\prod_{j=2}^{k}(1-\delta_j).$$ 
\end{lemma}
\begin{proof}
Notice that, $\lim\limits_{T\rightarrow \infty} a(k,l)=1$. The asymptotic performance of the proposed scheme can be analyzed as follows:
\begin{small}
\begin{eqnarray*}
\lim\limits_{T\rightarrow \infty}P_B^{iso}(k+1,i)&=& \displaystyle \sum_{l\in\{0,1\}} P_l \lim\limits_{T\rightarrow \infty} \left[ (1-b(k,l))  \left(\frac{a(k+1,l)}{a(k,l)}\right)P_B^{iso}(k,i,l)\right]\\
&\geq& (1-\delta_k)\displaystyle \sum_{l\in\{0,1\}} P_l \lim\limits_{T\rightarrow \infty} \left[P_B^{iso}(k,i,l)\right]\\
&=& \prod_{j=2}^{k}(1-\delta_j).
\end{eqnarray*}
\end{small}
\end{proof}
Notice that, the parallel network topology is a special case of the tree network topology with $K=1$. For $K=1$, our scheme can identify all the Byzantines with probability one because $\lim\limits_{T\rightarrow \infty}P_B^{iso}(1,i)=\lim\limits_{T\rightarrow \infty}\displaystyle \sum_{l\in\{0,1\}} P_l [a(1,l)]=1$. 
When $K>1$, we can choose $\eta_k$ appropriately such that Byzantines can be identified with a high probability. 

Next, to gain insights into the solution, we present some numerical results in Figure~\ref{pbiso} that corroborate our theoretical results. We consider a tree network with $K=5$ and plot $P_B^{iso}(k,i)$, $1\leq k\leq 5$, as a function of the time window $T$. 
We assume that the operating points $(P_d^k,P_{fa}^k)$, $1\leq k\leq 5$, for the nodes at different levels are given by $[(0.8,0.1),(0.75,0.1),(0.6,0.1),(0.65,0.1),(0.6,0.1)]$ and for anchor node $(P_d^A,P_{fa}^A)=(0.9,0.1)$. We also assume that the hypotheses are equi-probable, i.e., $P_0=P_1=0.5$, and the maximum isolation probability of honest nodes at level $k$ based solely on its data is constrained by $\delta_k=0.01,\forall k$. It can be seen from Figure~~\ref{pbiso} that in a span of only $T=25$ time windows, our proposed scheme isolates/identifies almost all the Byzantines in the tree network. 
\begin{figure}[t]
  \centering
    \includegraphics[height=2.5in, width=!]{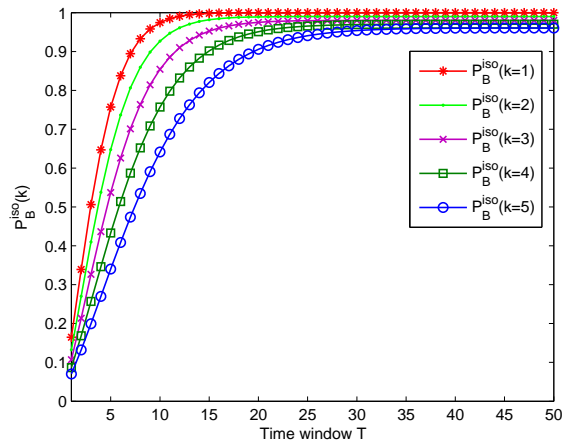}
    \vspace*{-0.1in}
    \caption{Isolation probability $P_B^{iso}(k,i)$ vs. time window $T$.} \label{pbiso}
  \vspace*{-0.15in}
\end{figure}
\section{Conclusion}
\label{sec6}
In this paper, we considered the problem of optimal Byzantine attacks on distributed detection mechanism in tree networks.
We analyzed the performance limit of detection
performance with Byzantines and obtained the optimal attacking strategies that minimize the detection error exponent. The problem was also studied from the network designer's perspective. It was shown that the optimal local detector is independent of the Byzantine's parameter. Next, we modeled
the strategic interaction between the FC and the attacker as a Leader-Follower (Stackelberg) game and attacker and defender (FC) equilibrium strategies were identified.
We also proposed a simple yet efficient scheme to identify Byzantines and analytically evaluated its performance. There are still many interesting questions that remain to be explored in the future work such
as analysis of the problem for arbitrary network topologies. The case where Byzantines collude in several groups (collaborate) to degrade
the detection performance can also be investigated.

\appendices
\section{Proof of Lemma~\ref{lem}}
\label{ap1}
To prove the lemma, we first show that any positive deviation $\epsilon\in(0,p]$ in flipping probabilities $(P_{1,0}^{k}, P_{0,1}^{k})=(p,p-\epsilon)$ will result in an increase in $D_k$. After plugging in $(P_{1,0}^{k}, P_{0,1}^{k})=(p,p-\epsilon)$ in \eqref{eq3} and \eqref{eq4}, we get
\begin{eqnarray}
\pi_{1,0}^k&=&[\beta_{1,0}^{k-1}(1-P_{fa}^{k})+(1-\beta_{0,1}^{k-1})P_{fa}^{k}]+[\alpha_k(p-P_{fa}^k(2p-\epsilon))+P_{fa}^k]\label{eeq2}\\
\pi_{1,1}^k&=&[\beta_{1,0}^{k-1}(1-P_d^{k})+(1-\beta_{0,1}^{k-1})P_d^{k}]+[\alpha_k(p-P_d^k(2p-\epsilon))+P_d^k].\label{eeq1}
\end{eqnarray}
Now we show that $D_k$ is a monotonically increasing function of the parameter $\epsilon$ or in other words, $\dfrac{dD_k}{d\epsilon}>0$.
\begin{eqnarray}
\dfrac{dD_k}{d\epsilon}&=&
\pi_{1,0}^k\left(\dfrac{\pi_{1,0}^{k'}}{\pi_{1,0}^k}-\dfrac{\pi_{1,1}^{k'}}{\pi_{1,1}^k}\right)+\pi_{1,0}^{k'} \log \dfrac{\pi_{1,0}^k}{\pi_{1,1}^k}\nonumber\\
&+&(1-\pi_{1,0}^k)\left(\dfrac{\pi_{1,1}^{k'}}{1-\pi_{1,1}^k}-\dfrac{\pi_{1,0}^{k'}}{1-\pi_{1,0}^k}\right)
-\pi_{1,0}^{k'}\log \dfrac{1-\pi_{1,0}^k}{1-\pi_{1,1}^k}\label{con1}
\end{eqnarray}
where $\dfrac{d\pi_{1,1}^k}{d\epsilon}=\pi_{1,1}^{k'}=\alpha_k P_d^k$ and $\dfrac{d\pi_{1,0}^k}{d\epsilon}=\pi_{1,0}^{k'}=\alpha_k P_{fa}^k$. After rearranging the terms in the above equation, the condition $\dfrac{dD_k}{d\epsilon}>0$ becomes 
\begin{eqnarray}
\dfrac{1-\pi_{1,0}^k}{1-\pi_{1,1}^k}+\dfrac{P_{fa}^k}{P_{d}^k}\log\dfrac{\pi_{1,0}^k}{\pi_{1,1}^k}>\dfrac{\pi_{1,0}^k}{\pi_{1,1}^k}+\dfrac{P_{fa}^k}{P_{d}^k}\log\dfrac{1-\pi_{1,0}^k}{1-\pi_{1,1}^k}.\label{log}
\end{eqnarray}
Since $P_d^k>P_{fa}^k$ and $\beta_{\bar{x},x}^k<0.5$, $\pi_{1,1}^k>\pi_{1,0}^k$. It can also be proved that $\dfrac{P_{d}^k}{P_{fa}^k}\dfrac{\pi_{1,0}^k}{\pi_{1,1}^k}>1$. Hence, we have

\begin{eqnarray}
&&
1+(\pi_{1,0}^k-\pi_{1,1}^k)<\dfrac{P_{d}^k}{P_{fa}^k}\dfrac{\pi_{1,0}^k}{\pi_{1,1}^k} \nonumber\\ \nonumber
&\Leftrightarrow&
(\pi_{1,0}^k-\pi_{1,1}^k)[1+(\pi_{1,0}^k-\pi_{1,1}^k)]>\dfrac{P_{d}^k}{P_{fa}^k}\dfrac{\pi_{1,0}^k}{\pi_{1,1}^k} (\pi_{1,0}^k-\pi_{1,1}^k)\nonumber\\
&\Leftrightarrow&
(\pi_{1,0}^k-\pi_{1,1}^k)\left[\dfrac{1+(\pi_{1,0}^k-\pi_{1,1}^k)}{\pi_{1,0}^k(1-\pi_{1,1}^k)}\right]>\dfrac{P_{d}^k}{P_{fa}^k}\dfrac{\pi_{1,0}^k}{\pi_{1,1}^k}\left[\dfrac{\pi_{1,0}^k-\pi_{1,1}^k}{\pi_{1,0}^k(1-\pi_{1,1}^k)}\right]\nonumber\\
&\Leftrightarrow&
(\pi_{1,0}^k-\pi_{1,1}^k)\left[\dfrac{1}{1-\pi_{1,1}^k}+\dfrac{1}{\pi_{1,0}^k}\right]>\dfrac{P_{d}^k}{P_{fa}^k}\left[\dfrac{\pi_{1,0}^k-\pi_{1,0}^k\pi_{1,1}^k+\pi_{1,0}^k\pi_{1,1}^k-\pi_{1,1}^k}{\pi_{1,1}^k(1-\pi_{1,1}^k)}\right]\nonumber\\
&\Leftrightarrow&
\left[\dfrac{1-\pi_{1,1}^k-(1-\pi_{1,0}^k)}{1-\pi_{1,1}^k}+\dfrac{(\pi_{1,0}^k-\pi_{1,1}^k)}{\pi_{1,0}^k}\right]>\dfrac{P_{d}^k}{P_{fa}^k}\left[\dfrac{\pi_{1,0}^k}{\pi_{1,1}^k}-\dfrac{1-\pi_{1,0}^k}{1-\pi_{1,1}^k}\right]\nonumber\\
&\Leftrightarrow&
\dfrac{1-\pi_{1,0}^k}{1-\pi_{1,1}^k}+\dfrac{P_{fa}^k}{P_{d}^k}
\left(1-\dfrac{\pi_{1,1}^k}{\pi_{1,0}^k}\right)
>\dfrac{\pi_{1,0}^k}{\pi_{1,1}^k}+\dfrac{P_{fa}^k}{P_{d}^k}
\left(\dfrac{1-\pi_{1,0}^k}{1-\pi_{1,1}^k}-1\right).\label{eq-1}
\end{eqnarray}

To prove that \eqref{log} is true, we apply the logarithm inequality $(x-1)\geq\log x \geq \dfrac{x-1}{x}$, for $x>0$ to \eqref{eq-1}. First, let us assume that $x=\dfrac{\pi_{1,0}^k}{\pi_{1,1}^k}$. Now using the logarithm inequality we can show that $\log\dfrac{\pi_{1,0}^k}{\pi_{1,1}^k}\geq 1-\dfrac{\pi_{1,1}^k}{\pi_{1,0}^k}$. Next, let us assume that $x=\dfrac{1-\pi_{1,0}^k}{1-\pi_{1,1}^k}$. Now using the logarithm inequality it can be shown that $\left[\dfrac{1-\pi_{1,0}^k}{1-\pi_{1,1}^k}-1\right] \geq \log \dfrac{1-\pi_{1,0}^k}{1-\pi_{1,1}^k}$. Using these results and \eqref{eq-1}, one can prove that condition \eqref{log} is true.

Similarly, we can show that any non zero deviation $\epsilon\in(0,p]$ in flipping probabilities $(P_{1,0}^{k}, P_{0,1}^{k})=(p-\epsilon,p)$ will result in an increase in $D_k$, i.e.,
$\dfrac{dD_k}{d\epsilon}>0$,
 or
\begin{eqnarray}
\dfrac{\pi_{1,0}^k}{\pi_{1,1}^k}+\dfrac{1-P_{fa}^k}{1-P_{d}^k}\log\dfrac{1-\pi_{1,0}^k}{1-\pi_{1,1}^k}>\dfrac{1-\pi_{1,0}^k}{1-\pi_{1,1}^k}+\dfrac{1-P_{fa}^k}{1-P_{d}^k}\log\dfrac{\pi_{1,0}^k}{\pi_{1,1}^k}.\label{log1}
\end{eqnarray} 
Since $P_d^k>P_{fa}^k$ and $\beta_{\bar{x},x}^k<0.5$, $\pi_{1,1}^k>\pi_{1,0}^k$. It can also be proved that $\dfrac{1-\pi_{1,0}^k}{1-\pi_{1,1}^k}<\dfrac{1-P_{fa}^k}{1-P_{d}^k}$. Hence, we have

\begin{small}
\begin{eqnarray}
&&
\dfrac{1-\pi_{1,0}^k}{1-\pi_{1,1}^k}<\dfrac{1-P_{fa}^k}{1-P_{d}^k}\left[1-(\pi_{1,0}^k-\pi_{1,1}^k)\right]\\
&\Leftrightarrow&
\dfrac{1-\pi_{1,0}^k}{\pi_{1,1}^k(1-\pi_{1,1}^k)}<\dfrac{1-P_{fa}^k}{1-P_{d}^k}\left[\dfrac{1-(\pi_{1,0}^k-\pi_{1,1}^k)}{\pi_{1,1}^k}\right]\nonumber\\
&\Leftrightarrow&
\dfrac{1}{\pi_{1,1}^k(1-\pi_{1,1}^k)}<\dfrac{1-P_{fa}^k}{1-P_{d}^k}\left[\dfrac{1-(\pi_{1,0}^k-\pi_{1,1}^k)}{\pi_{1,1}^k(1-\pi_{1,0}^k)}\right]\nonumber\\
&\Leftrightarrow&
\dfrac{1}{\pi_{1,0}^k-\pi_{1,1}^k}\left[\dfrac{\pi_{1,0}^k-\pi_{1,0}^k\pi_{1,1}^k+\pi_{1,0}^k\pi_{1,1}^k-\pi_{1,1}^k}{\pi_{1,1}^k(1-\pi_{1,1}^k)}\right]<\dfrac{1-P_{fa}^k}{1-P_{d}^k}\left[\dfrac{1-(\pi_{1,0}^k-\pi_{1,1}^k)}{\pi_{1,1}^k(1-\pi_{1,0}^k)}\right]\nonumber\\
&\Leftrightarrow&
\dfrac{1}{\pi_{1,0}^k-\pi_{1,1}^k}\left[ \dfrac{\pi_{1,0}^k}{\pi_{1,1}^k}-\dfrac{1-\pi_{1,0}^k}{1-\pi_{1,1}^k}\right]<\dfrac{1-P_{fa}^k}{1-P_{d}^k}\left[\dfrac{1}{\pi_{1,1}^k}+\dfrac{1}{1-\pi_{1,0}^k}\right]\label{eq-7}\\
&\Leftrightarrow&
\dfrac{\pi_{1,0}^k}{\pi_{1,1}^k}-\dfrac{1-\pi_{1,0}^k}{1-\pi_{1,1}^k}>\dfrac{1-P_{fa}^k}{1-P_{d}^k}\left[\dfrac{\pi_{1,0}^k-\pi_{1,1}^k}{\pi_{1,1}^k}+\dfrac{\pi_{1,0}^k-\pi_{1,1}^k}{1-\pi_{1,0}^k}\right]\\
&\Leftrightarrow&
\dfrac{\pi_{1,0}^k}{\pi_{1,1}^k}-\dfrac{1-\pi_{1,0}^k}{1-\pi_{1,1}^k}>\dfrac{1-P_{fa}^k}{1-P_{d}^k}\left[\dfrac{\pi_{1,0}^k-\pi_{1,1}^k}{\pi_{1,1}^k}+\dfrac{1-\pi_{1,1}^k-(1-\pi_{1,0}^k)}{1-\pi_{1,0}^k}\right]\nonumber\\
&\Leftrightarrow&
\dfrac{\pi_{1,0}^k}{\pi_{1,1}^k}+\dfrac{1-P_{fa}^k}{1-P_{d}^k}
\left[1-\dfrac{1-\pi_{1,1}^k}{1-\pi_{1,0}^k}\right]
>\dfrac{1-\pi_{1,0}^k}{1-\pi_{1,1}^k}
+\dfrac{1-P_{fa}^k}{1-P_{d}^k}\left[\dfrac{\pi_{1,0}^k}{\pi_{1,1}^k}-1\right].\label{eq-11}
\end{eqnarray}
\end{small} 
To prove that \eqref{log1} is true, we apply the logarithm inequality $(x-1)\geq\log x \geq \dfrac{x-1}{x}$, for $x>0$ to \eqref{eq-11}. First, let us assume that $x=\dfrac{1-\pi_{1,0}^k}{1-\pi_{1,1}^k}$. Now using the logarithm inequality we can show that $\log\dfrac{1-\pi_{1,0}^k}{1-\pi_{1,1}^k}\geq 1-\dfrac{1-\pi_{1,1}^k}{1-\pi_{1,0}^k}$. Next, let us assume that $x=\dfrac{\pi_{1,0}^k}{\pi_{1,1}^k}$. Now using the logarithm inequality it can be shown that $\left[\dfrac{\pi_{1,0}^k}{\pi_{1,1}^k}-1\right] \geq \log \dfrac{\pi_{1,0}^k}{\pi_{1,1}^k}$. Using these results and \eqref{eq-11}, one can prove that condition \eqref{log1} is true.

\bibliographystyle{IEEEtran}
\bibliography{Conf,Book,Journal}

\end{document}